\documentclass{article}

\usepackage{color}

\usepackage{amsmath}
\usepackage{algorithm}
\usepackage{algorithmic}
\newcommand\numberthis{\addtocounter{equation}{1}\tag{\theequation}}
\DeclareMathOperator*{\argmax}{argmax}

\usepackage{graphicx}
\usepackage{amsthm}

\usepackage{subfigure}
\usepackage{mathrsfs}
\usepackage{color}
\usepackage{enumerate}
\usepackage{balance}
\usepackage{amssymb}
\usepackage{multirow}
\usepackage{fancyvrb}
\usepackage{url}
\usepackage[usenames,dvipsnames]{xcolor}

\newtheorem{theorem}{Theorem}

\newtheorem{corollary}{Corollary}
\newtheorem{definition}{Definition}

\newcommand{\bpara}[2]{\noindent\textbf{#1~~}{#2}}

\begin{document}
%

\title{Achieving Energy Efficiency in Cloud Brokering}
\author{Lin Wang$^\ast$, Lei Jiao$^\dagger$, Mateusz Guzek$^\ast$, \\Dzmitry Kliazovich$^\ast$, Pascal Bouvry$^\ast$ \vspace{0.15cm} \\
$^\ast$ SnT, University of Luxembourg \\$^\dagger$ Bell Labs
}
\date{11 April 2016}

\maketitle
\begin{abstract}
The proliferation of cloud providers has brought substantial interoperability complexity to the public cloud market, in which cloud brokering has been playing an important role. However, energy-related issues for public clouds have not been well addressed in the literature. In this paper, we claim that the broker is also situated in a perfect position where necessary actions can be taken to achieve energy efficiency for public cloud systems, particularly through job assignment and scheduling. We formulate the problem by a mixed integer program and prove its NP-hardness. Based on the complexity analysis, we simplify the problem by introducing admission control on jobs. In the sequel, optimal job assignment can be done straightforwardly and the problem is transformed into improving job admission rate by scheduling on two coupled phases: data transfer and job execution. The two scheduling phases are further decoupled and we develop efficient scheduling algorithm for each of them. Experimental results show that the proposed solution can achieve significant reduction on energy consumption with admission rates improved as well, even in large-scale public cloud systems.

\end{abstract}




\section{Introduction}
\label{sec:intro}

Cloud computing has been proven to be one of the most successful computing models in the past decade. To keep pace with the proliferation of online cloud services, enormous number of mega data centers have been built widely. However, complex interoperability between public cloud providers and tenants is becoming a big obstacle for the fast and flexible deployment and operating of new cloud services. Serving as an intermediate entity between cloud providers and tenants, cloud brokers has brought about many great benefits to the cloud market, among which flexibility possesses its prominence. In this new model, a Cloud Service Broker (CSB) rents either resources or services from multiple cloud providers and then resell them to tenants. Depending on tenants' needs, a CSB may pack different services or integrate them with its own added-value services such as data encryption. The brokering model directly leads to the fact that data storage or processing would need to be coordinated among a set of geographically distributed data centers from multiple cloud providers.

Apart from increased deployment flexibility, cloud brokering also brings another big opportunity for reducing the energy cost for public cloud providers. It is evident that the explosive expansion of data centers has resulted in a severe environmental concern over energy consumption or carbon footprint. To alleviate this situation, a large body of energy-efficient architectures or algorithms for single data centers or private clouds has been proposed (e.g., \cite{Islam-Colo-2015, Lin-DRS-2011, Wang-GreenDCN-2014}), but little attention has been paid on the energy efficiency of public clouds. In the context of cloud brokering, a broker can have a global view over not only the providers, but also the subscribed tenants, which enables a comprehensive understanding on resource demand and provision balance in the system. In the ideal case, the broker can choose the most appropriate data center site (with minimized energy cost) to execute each job demand from the tenants based on this global view. We notice that reducing the energy cost of data centers is not always aligned with the broker's interest of maximizing its own profit. However, the broker can be simply incentivized by a deliberate pricing policy from the providers \cite{Qiu-CSB-2015}. Our work will assume that such a pricing policy has been applied already.

\begin{figure}
\centering
\includegraphics[scale=0.65]{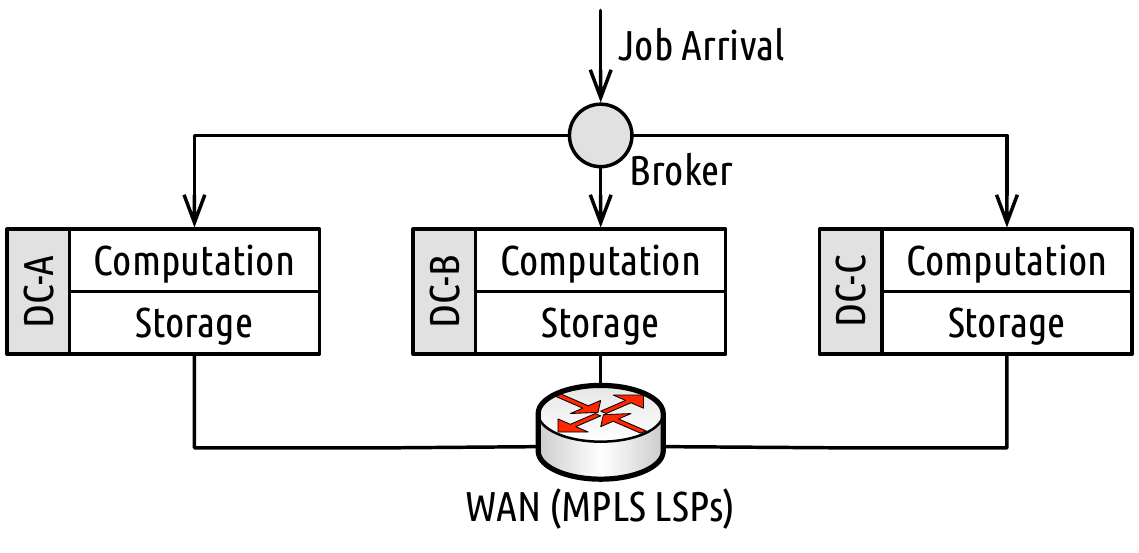}
\caption{\label{fig:framework}An overview of the proposed cloud brokering model. A broker is in charge of job dispatching, data movement, and job scheduling among a set of geo-distributed data centers.}
\end{figure}

In this paper, we investigate the problem of optimizing the energy cost of geographically distributed data centers from the broker's point of view. We consider an abstracted brokering model where a broker is in charge of the use of the hardware resources\footnote{For simplicity we only consider the dominant resource dimensions, i.e., computation, storage, and networking.} from a set of geo-distributed data centers and accept job requests from its subscribed tenants. An overview of the adopted cloud brokering model is depicted in Figure~\ref{fig:framework}. We restrict our attention on batch workloads (mostly based on data analytical stacks like Hadoop) such as web indexing and data mining which are dominant in large organizations today \cite{Facebook-data}. Additionally, these jobs are data intensive and are more tolerant to processing delay, providing more flexibility for us in job assignment and scheduling.\footnote{In contrast, those long-running, stateless services such as chat or three-tier web applications usually adopt a replicated execution model, where the very limited amount of data is simply replicated among multiple data center sites.} In this model, we assume that the data to be processed for each job has been pre-stored in one of the data centers by the tenant. Upon job arrivals, the broker has to decide \textit{the data center each job will be assigned to} and \textit{the execution order of the jobs that have been assigned to the same data ceter}. If the job is assigned to a different data center than where its associated data is stored, the data has to to transferred to the target data center, which inevitably incurs an overhead of networking cost. The objective is to minimize the overall cost (i.e., energy and networking costs) from all the data centers, while guaranteeing the deadline given to each job. Our model characterizes the real scenarios and to the best of our knowledge it has not been sufficiently covered in the literature.

To solve this problem, we first formalize the model and formulate the problem with a mixed integer program. Non-surprisingly, the problem can be proven to be NP-hard. Moreover, we show that regardless of optimality, even deciding whether the problem has a feasible solution or not is already NP-complete. This is due to the fact that job scheduling with deadline constraints is hard in general even under a single-machine setting. Then, based on the complexity results, we observe the necessity of introducing admission control (to block some jobs out if their deadlines cannot be met) while pursuing optimality in terms of energy cost in the system. Hereafter, the problem is transformed into maximizing the admission rate by job scheduling while maintaining the best cost-effectiveness by assigning jobs to `cheaper' data centers. While the optimal job assignment can be done straightforwardly, the problem of scheduling for maximized admission rate is still non-trivial. Finally, the scheduling problem is decoupled into two phases: data transfer and job execution, each of which is handled separately.

Our contributions can be summarized into the following three aspects: $i)$ We formulate the job assignment and scheduling problem into a mixed integer program and show its time complexity. $ii)$ We introduce admission control to simplify the problem and transform the problem into a novel yet more general problem -- Flow Shop Scheduling with Coupled Resources (FSS-CR). $iii$) We propose an efficient algorithm for FSS-CR and validate its performance by numerical simulations.

The remainder of this paper is organized as follows. Section~II summarizes related work. Section~III provides the system model, the problem formulation, and complexity analysis. Section~IV presents our algorithm design. Section~V validates the performance of the proposed algorithm by simulations. Section~VI discusses possible extensions and future directions and Section~VII concludes the paper.

\section{Related Work}
\label{sec:related}

There has been a large body of work on both cloud brokering and cost reduction in geo-distributed data center systems. In this section, we summarize the most representative ones and differentiate them from our study.

\bpara{Cloud brokering and federation.}{
Functioning as an intermediary between cloud providers and tenants, a cloud broker can largely simplify the deployment and improve the provisioning flexibility of cloud services for enterprises in the public cloud market. In the sequel, cloud brokering has ignited the focus on both designing sustainable pricing models and developing technical solutions for service interoperability \cite{Guzek-CSB-2015}. Similarly, cloud federation is based on the idea of a smart sharing of workloads among multiple cloud providers based on some mutual agreement. Cerroni \cite{Cerroni-DCC-2014} provides a comprehensive analytical model for joint dimensioning of shared network and data center capacity in a federate cloud by taking advantage of virtual machine live migration. }

\bpara{Energy efficiency in geo-distributed data centers.}{Ren \emph{et al.} \cite{Ren—GreFar-2012} propose GreFar, a provably-efficient online scheduling algorithm for batch job scheduling to achieve energy efficiency and fairness in geo-distributed data centers by exploring the benefit of electricity price variations across time and location. However, the cost for data movement is not included in their model. Xu \emph{et al.} \cite{Xu-TPDS-2015} explore the correlation between ambient temperature and cooling efficiency. By taking advantage of both time and geographical diversity of the temperature at individual locations, they advocate a joint optimization of request routing for interactive workloads and capacity allocation for batch workloads and develop a distributed algorithm for achieving cost effectiveness in geo-distributed data centers. Buchbinder \emph{et al.} \cite{Buchbinder-2011} study the problem of migrating jobs among data centers to achieve reduced electricity cost by exploiting the variation of electricity prices both temporally and geographically. They consider not only the energy expenses but also the bandwidth cost introduced by job migration. However, the assumption that the job migration time is negligible is not realistic as moving data between data centers, especially for data-intensive jobs, usually takes a remarkable amount of time, leading to a significant delay in job execution. Qiu \emph{et al.} \cite{Qiu-CSB-2015} propose a pricing model to enable broker's incentive to reduce providers' energy cost, based on which they develop demand allocation mechanisms to achieve the best energy efficiency while maintaining a high level of resource utilization.

In contrast, our model takes into account data movement in terms of both time and cost. We try to reduce the joint cost of energy and networking and to ensure Service Level Agreement (SLA) by enforcing hard deadlines for job completion through joint job assignment and scheduling.}

\section{The Model}
\label{sec:model}

In this section, we present the system model, formulate the problem and carry out complexity analysis for the problem.

\subsection{System Model}

We adopt a discrete time model where the length of a time slot matches the time scale at which job dispatching and scheduling decisions are made, e.g., hourly. Different from long-running, stateless workloads, back-end batch workloads usually have better tolerance on processing delay, as long as they can be accomplished by a relatively loose deadline. This deadline is usually used to ensure the availability of the processing result before certain transactions.

We consider a cloud broker (or a cloud federation) who manages the hardware resources from a set of $m$ data centers in distinct geographical regions, denoted by $\mathbf{S} = \{S_1, S_2,...,S_m\}$. Each of the data center $S_i \in \mathbf{S}$ is equipped with a certain number of servers and the maximum computing capacity is captured by $C_i$.\footnote{The computing capability of a data center is quantified by the number of instructions executed per second (in MIPS).} All the data centers are connected by a Wide Area Network (WAN) where MPLS Label Switching Paths (LSPs) are pre-established between data center pairs for direct communication, as depicted in Figure~\ref{fig:framework}. We assume no bottleneck in the network core, i.e., the only bottleneck exists in the connection between each data center and the core. For simplicity we abstract the network as one non-blocking switch with \emph{heterogeneous} transmission rates on ports (as shown in Figure~\ref{fig:wan-switch}) and thus, the downlinks and uplinks are the only sources of contention. Note that the heterogeneity brings novelty, as well as new challenges, to the problem. The download and upload network bandwidth of each data center $S_i$ is upper-bounded by $B_i^{in}$ and $B_i^{out}$, respectively.

\begin{figure}
\centering
\includegraphics[scale=0.65]{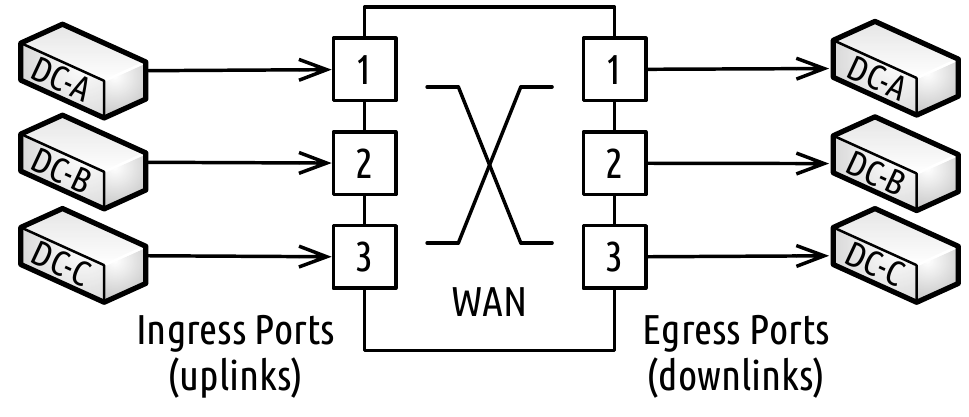}
\caption{\label{fig:wan-switch} A non-blocking switch abstraction for the WAN spanning three geo-distributed data centers. }
\end{figure}

The broker receives job requests from the various tenants and determines the combination of choices for the following factors: the data center to assign each job and the order in which jobs are scheduled for both cross-data center data transfer and job execution at each data center. We call this problem ASCO -- Assigning and Scheduling for Cost Optimization.

\subsection{Job Characterization}

We are given a set of $n$ jobs $\mathbf{J} = \{J_1, J_2,..., J_n\}$ that need to be executed in the aforementioned set of geo-distributed data centers. These jobs are requested by cloud tenants for data processing. Each of the job is associated with a piece of data that has already been stored in one of the data centers following some data storage policies \cite{Liu-Storage-2015}. The parameters for describing each job $J_j \in \mathbf{J}$ are given by a five-tuple $\langle$$a_j$, $b_j$, $l_j$, $d_j$, $S_{d, j}$$\rangle$, where the elements of the five-tuple are defined in Table~\ref{tb:notation}.
\begin{table}
\centering
\caption{\label{tb:notation}Job Parameters}
\begin{tabular}{c|l} \hline
	$a_j$, $b_j$ & arriving time and deadline\\ \hline
	$l_j$ & computation workload \\ \hline
	$d_j$ & volume of input data \\ \hline
	$S_{d,j}$ & data center that currently stores the data \\
\hline\end{tabular}
\end{table}
Each job has to be completed before its deadline. Note that in order to achieve resource efficiency and to reduce energy cost, it is not necessary to have a job to be executed in the same data center that the data for the job resides. As a result, if a job is assigned to a data center that fortunately stores its associated data, it will be scheduled and executed directly; otherwise, the data for the job has to be first transferred to the assigned data center and then, the data processing is carried out. 

The broker will first need to make decisions on which data center to assign each job. We denote $x_{i,j} \in \{0, 1\}$ as a decision variable indicating whether job $J_j$ is assigned to data center $S_i$, where
\begin{equation}
x_{i,j} = \left\{
  \begin{array}{ll} 
   1 & \text{if } J_j \text{ is assigned to } S_i, \\
   0 & \text{otherwise}.
  \end{array}
\right.
\end{equation}
We also denote $y_j \in \{0, 1\}$ as an indicator for whether the assigned data center $S_i$ for job $J_j$ is the same as the one $S_{d, j}$ that stores the data for this job, i.e., $y_j = 1$ if $S_i = S_{d,j}$; $y_j = 0$ otherwise.

\subsection{Total Cost}

We denote by $H$ the combined cost for processing a set of jobs $\mathbf{J}$. The cost $H_i$ we consider here for each data center $S_i$ consists in two parts: energy cost $E_i$ and network cost $N_i$. We assume that the energy cost is linearly related to the workload and the electricity price given by $P_i$ (i.e., energy cost per unit of computation) at site $S_i$. The electricity cost varies in different sites as in different geographical locations the electricity generation cost can be different. Note that we assume a simplified model for tractability but it can be further calibrated to a more sophisticated model by taking into account also cooling efficiency due to temperature fluctuation in both time and location dimensions at different sites \cite{Xu-TPDS-2015}.

The network cost $N_i$ is proportional to the volume of data being transferred across data centers which covers both upload and download traffic. To stay generic we assume that the prices for sending and receiving unit of data at the same site are not identical and both prices are also heterogeneous among the sites. Denote by $Q_i^{in}$ and $Q_i^{out}$ as the download and upload network price at site $S_i$, respectively. We note that there are also network cost associated with the arrival of jobs into the system (e.g., requests from the tenant) and leaving the system (e.g., delivering the result to the tenant). However, those costs are usually small and are independent with the job assignment and scheduling. Thus, they are omitted from our model.

For a given job $J_j$ being processed in data center $S_i$, the energy cost can be expressed simply by $P_i \cdot l_j$, while the network cost is captured by $y_j \cdot d_j \cdot (Q_{d,j}^{out} + Q_i^{in})$. Denoting by $\mathbf{J}_i$ the set of jobs that are assigned to data center $S_i$, i.e.,
\begin{equation}
\mathbf{J}_i = \{ J_j~|~J_j \in \mathbf{J} \wedge x_{i,j} = 1 \},
\end{equation}
the total cost $H_i$ at site $S_i$ is given by
\begin{equation}
H_i = \sum_{J_j \in \mathbf{J}_i} \left( \sum_{S_i \in \mathbf{S}} x_{i,j} \cdot P_i \cdot l_j + y_j \cdot d_j \cdot (Q_{d,j}^{out} + Q_i^{in}) \right).
\end{equation}
The total cost $H$ of the system is given by the sum of the costs at all sites in the system, i.e., $H = \sum_{S_i \in \mathbf{S}} H_i$, which we aim to minimize.

\subsection{Problem Formulation}

The goal of the ASCO problem is to assign jobs to proper data centers such that the total cost $H$ is minimized, while the deadlines of all the jobs can be respected. We denote by $t_0$ and $t_1$ the earliest arriving time and the latest deadline of the jobs in set $\mathbf{J}_i$, i.e., $t^a_i = \min\{a_j~|~J_j \in \mathbf{J}_i\}$ and $t_i^b = \max\{b_j~|~J_j \in \mathbf{J}_i\}$. Given an arbitrary non-empty subset $\widehat{\mathbf{J}_i}$ of $\mathbf{J}_i$, we denote by $\widehat	{t^a_i}$ and $\widehat{t^b_i}$ as the corresponding earliest arriving time and the latest deadline for the jobs in $\widehat{\mathbf{J}_i}$. The minimum total computation time for all the jobs in $\widehat{\mathbf{J}_i}$ is given by
\begin{equation}
\sum_{J_j \in \widehat{\mathbf{J}_i}} \frac{l_j}{C_i}, \label{comp}
\end{equation}
while the minimum total communication time taken by transferring the data for the jobs in $\widehat{\mathbf{J}_i}$ can be represented by
\begin{equation}
\sum_{J_j \in \widehat{\mathbf{J}_i}} y_j \cdot \frac{d_j}{\min(B_{d,j}^{out}, B_i^{in})}. \label{comm}
\end{equation}
Putting everything together, the optimization problem ASCO can be further formulated as the following integer program.
\begin{equation}
\begin{aligned}
&(P_1)~~~~\min  H  \\
\text{subject to}  \\
 &(\ref{comp}) + (\ref{comm}) \leq \widehat{t_i^b} - \widehat{t_i^a} & \forall \widehat{\mathbf{J}_i} \subseteq \mathbf{J}_i  \\
&\sum_{S_i \in \mathbf{S}} x_{i, j} = 1 & \forall J_j \in \mathbf{J} \\
& x_{i,j} \in \{0, 1\} & \forall S_i \in \mathbf{S}, \forall J_i \in \mathbf{J} \\
\end{aligned}
\nonumber
\end{equation}
The first inequality ensures that all the jobs can be completed before their deadlines at the given site. The second constraint is to force that every job is assigned to one and only one site, while the last constraint is the binary constraint for the decision variable $x_{i,j}$. Note that $y_j$ is a variable whose value is totally associated with $x_{i,j}$ so no decision is needed on $y_j$.

\subsection{Complexity Analysis}
We now analyze the complexity of the afore-defined ASCO problem. We first introduce a new problem called F-ASCO, which is defined as deciding whether there exists a feasible solution to the ASCO problem, \textit{regardless of} the total cost. We show the complexity of the F-ASCO problem in the following.
\begin{theorem}
The F-ASCO problem is NP-complete.
\end{theorem}
\begin{proof}
The goal of the F-ASCO problem is deciding whether there is an assignment of jobs to sites, together with a schedule of the jobs at all sites, such that the deadlines of all the jobs are respected in the ASCO problem, regardless of the total cost. The proof can be conducted by a polynomial time reduction from the classical Minimum Makespan Scheduling (MMS)  problem whose decision version is NP-complete even if there are only two identical machines \cite{Garey-MMS-1979}. 

We start from an MMS instance where we are given a set of identical machines indexed by the set $M = \{1, ..., m\}$ and a set of jobs indexed by the set $J = \{1,...,n\}$ to be assigned to the machines. Each job contains a certain workload $w_j$ ($j \in [1, n]$) to be processed. The goal is to assign and process the jobs on the machines and the objective is to minimize the makespan, i.e., the maximum completion time of the machines. We denote by $OPT_{0}$ the minimum makespan that can be achieved.

From the above MMS instance we now construct an instance for the F-ASCO problem. We assume each machine in $M$ represents a data center site $S_i$ and we have in total $n$ sites given by the set $\mathbf{S}$ Each job $j$ in $J$ represents a job $J_j \in \mathbf{J}$ where $\mathbf{J}$ denotes the set of jobs for the F-ASCO instance. For all the jobs in $\mathbf{J}$, we assume they arrive at the same time and have the same deadline as $OPT_0$. We also assume that the network bandwidths at every site are infinite and the time used for data transmission is thus negligible. The question we need to answer in the F-ASCO problem instance then becomes to correctly decide whether there exists a schedule for the jobs to data center sites such that all the jobs can be completed within $OPT_0$. It is easy to confirm that our answer to the F-ASCO instance is YES if, and only if, we solve the MMS instance and find its optimal solution. That completes the proof.
\end{proof}

The following result then can be directly derived from the above theorem as pursuing optimality is obviously one step further than identifying feasibility.

\begin{corollary}
The ASCO problem is NP-hard.
\end{corollary}


\section{The Algorithm}
\label{sec:algo}

Since we cannot efficiently have any guarantee on the existence of feasible solutions, we introduce \textit{admission control}, which is used to block some job requests when there is no enough computation or network resource to serve them. The problem is then relaxed and the optimal cost can be achieved at the risk that there might be jobs that cannot be accommodated. In the following, we will show how to achieve the best cost-effectiveness by site selection for jobs and how to improve admission rate by carrying out a well-designed scheduling algorithm for joint cross-site data transfer and job execution.

\subsection{Site Selection}

The incentive to transfer the data for a job from one site to another is to reduce the total cost for processing the job. As a result, the following two necessary (yet not sufficient) conditions have to be met: (1) Data transfer is possible subject to time limit; (2) The total cost for processing the job is reduced as a result of the data transfer.

Assume there is a job $J_j \in \mathbf{J}$ and the data for this job is stored at site $S_{d,j}$. According to the above two conditions, a candidate site $S_c \in \mathbf{S} \backslash S_{d,j}$ that can host job $J_j$ has to satisfy the following inequalities.
\begin{align*}
\frac{l_j}{C_c} + \frac{d_j}{\min (B_{d,j}^{out}, B_c^{in})} & \leq  b_j - a_j \label{ineqn:deadline-constraint}\numberthis\\
l_j \cdot P_c + d_j \cdot (Q_{d,j}^{out} + Q_c^{in}) & <  l_j \cdot P_i \label{ineqn:cost-constraint}\numberthis
\end{align*}
Based on the above two conditions, for each job $J_j$ we carry out a screening process, which aims at removing the sites that cannot host the job. We denote by $\mathbf{S}_j$ the set of the valid candidate sites for job $J_j$. If $\mathbf{S}_j = \emptyset$, the job will be by default assigned to site $S_{d,j}$; otherwise we choose the site from all the candidate sites $\mathbf{S}_j$ that gives the minimized cost for executing the job (including the cost for data transfer). We denote the chosen site by $S_{p,j}$. The assignment of jobs will be completed when we finish repeating the above process for every job.



\subsection{Two-phase Job Scheduling}

Once we have decided the site that each job will be preferably assigned, the problem becomes how to schedule the data transfer and job execution for the jobs at each site. Having in mind that there might be jobs that are not admitted into the system, our implicit objective for the scheduling would be to accommodate as more jobs as possible in order to maintain higher tenant satisfaction.

We first describe a similar problem that has been widely studied in traditional job scheduling literature: We are given three sets of parallel machines, denoted by $\mathbf{A}, \mathbf{B}$ and $\mathbf{C}$, respectively. Note that the machines in each set can be heterogeneous with non-uniform processing capabilities. We are also given a set of jobs, each of which consists of three operations that have to be carried out on the three sets of machines stage by stage, respectively. The objective of the problem is to minimize the total completion time of all the jobs. If the operations of each job can be flexibly assigned to any of the machines in the corresponding class, the problem is called {\em Flexible Flow Shop Scheduling with Parallel Machines} and it has been shown to be strongly NP-hard even when the machines are uniform \cite{Kyparisis-OR-2006}. 

Our problem inherits the same problem structure but differs in the following two aspects: (1) The machine for carrying out the operation of each job in every stage is fixed. (2) The machines for carrying out the first two operations of each job are coupled, i.e., they will be occupied at the same time. The problem under the two constraints is still NP-hard. The proof can be straightforwardly conducted by a reduction from the traditional flow shop scheduling problem, which is known to be NP-hard with at least two machines (three machine sets in our case). Our objective, however, instead of minimizing the total job completion time, is to maximize the number of jobs that could be completed before given deadlines. We call this problem {\em Flow Shop Scheduling with Coupled Resources (FSS-CR)}.
The term {\em coupled resources} implies that the scheduling decisions for the first two operations of each job have to be jointly done at the same time.


\subsection{FSS-CR}

The high complexity of the problem implies a vast searching space, and we thus try to explore some useful insights and then design efficient heuristics that could generate comparably good results within very short time. Under the objective of accommodating as more jobs as possible, our algorithm aims at reducing the time wasted at every stage due to resource contention.

While sharing a common deadline, the decision making process for scheduling each job can be divided into two independent phases: \texttt{dtrans} (data transfer, the first two stages) and \texttt{comp} (computation, the third stage). We notice that once we fix the scheduling for the jobs in phase \texttt{comp}, the deadlines for phase \texttt{dtrans} of jobs can be accordingly determined. This is achievable as the set of jobs that will be processed at each site is already knowable after the site selection process described before. The problem is then decomposed into two sub-problems, i.e., scheduling for \texttt{comp} and scheduling for \texttt{dtrans}.

\subsubsection{Scheduling for Computation}
As in phase \texttt{comp} jobs will be processed independently at each site, we focus on an arbitrary site $S_i \in \mathbf{S}$. The main idea behind this is to ensure that all the jobs at this site can be completed before their deadlines and every job starts its computation phase as late as possible to make time for the data transfer phase as resource contention is more severe in data transfer due to resource coupling. To this end, we define an auxiliary problem called {\em Reverse-Job Scheduling (RJS)}. Given a set of jobs with deadline constraints to be processed on a single machine, the goal of the RJS problem is to decide the order of jobs to be processed such that the average starting time of all the jobs is maximized.

We notice that the RJS problem can be transformed into a single-machine job scheduling problem by treating job deadlines as arriving times and reversing the job execution from the end to the beginning. As a result, the problem becomes that given a set of jobs that arrive at a single machine, we design a schedule for the jobs such that the average completion time is minimized. A simple example of the transformation as well as the comparison of possible schedules (here we consider First Come First Serve (FCFS), Earliest Deadline First (EDF) and Shortest Remaining Time First (SRTF)) is shown in Figure~\ref{fig:fss-cr-comp}. Through the comparison we observe that the preemptive scheduling policy SRTF gives the minimal average completion time compared to non-preemptive policies FCFS and EDF. Consequently, the total time saved in the \texttt{comp} phase would be maximized.


\begin{figure}
\centering
\includegraphics[scale=0.8]{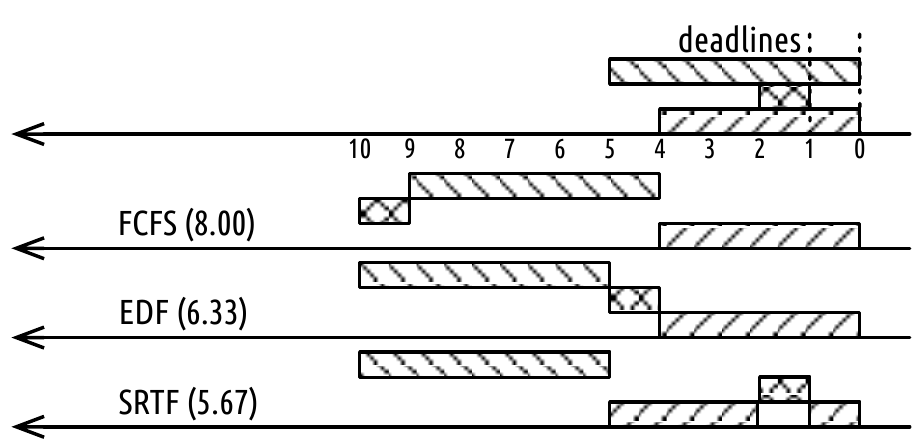}
\caption{\label{fig:fss-cr-comp}Transforming the RJS problem into a single-machine scheduling problem and the comparison of average job completion times given by different scheduling policies (i.e., FIFO, EDF, and SRTF).}
\end{figure}

\subsubsection{Scheduling for Data Transfer}

The starting time determined for each job in phase \texttt{comp} will provide a strict deadline for each job in phase \texttt{dtrans}. Given this deadline, the problem of scheduling for data transfer becomes a network flow scheduling problem: each job $J_j \in \mathbf{J}$ represents a flow $F_j$ with size of $d_j$ from site $S_{d,j}$ to site $S_{p,j}$, with arriving time $a_j$ and deadline $b_j$. Our goal is to schedule the flow transmissions (i.e., sending or receiving) on a semi-clos network while guaranteeing flow deadlines. Our solution to this problem is based on two iterative processes where the first process is admission control based on pruning, while the second one is scheduling. Before presenting the algorithm, we first carry out a flow size normalization process. 

\bpara{Flow Size Normalization.}Normalization is necessary for computing the most intensive time interval, as the maximal transmission rate for each flow can be different due to the fact that network nodes (i.e., sites) we assumed here can have different capacities. The normalization process is straightforward, i.e., assuming the capacity of each network port as one. To this end, for each flow we divide its size by the maximal transmission rate it can achieve when being transmitted on the network. More formally, the normalized size of flow $F_j$ is given by 
\begin{equation}
|F_j| = \frac{d_j}{\min(B_{d,j}^{out}, B_{p,j}^{in})}
\end{equation}
where $B_{d,j}^{out}$ and $B_{p,j}^{in}$ are the egress bandwidth of site $S_{d,j}$ and the ingress bandwidth of site $S_{p,j}$, respectively.

\bpara{Admission Control.}
The admission control process is designed to prune the flow set by removing the flows that will not be able to be completed, such that the feasibility of the transmissions of the residual flows can be guaranteed. Denote by $\mathbf{F}_i$ the set of flows that will be routed through site $S_i$. We first provide the following definition.
\begin{definition}
The intensity of a site $S_i$ in a given time interval $I = [a, b]$ is defined as the average normalized amount of data to be transmitted by $S_i$ in this interval, i.e.,
\begin{equation}
\delta(S_i, I) = \frac{\sum_{[a_j, b_j] \subseteq{[a, b]} \wedge F_j \in \mathbf{F}_i} |F_j|}{a \sim b}
\end{equation}
where $a \sim b$ denotes the total time in which site $S_i$ is free.
\end{definition}

It is intuitive that $\delta(S_i, I)$ has to satisfy $\delta(S_i, I) \le 1$, meaning that the maximal intensity is constrained by the normalized capacity of each site; otherwise there will be flows that cannot meet their deadlines. Our design for admission control is based the inequality $\max\{\delta(S_i, I)\} \le 1$.

\begin{definition}
For a given site $S_i$, a time interval $I = [a, b]$ is defined as a critical interval if it maximizes $\delta(S_i, I)$.
\end{definition}

\begin{definition}
The most critical time interval $I^* = [a^*, b^*]$ is defined as the time interval that maximizes $\delta(S_i^*, I^*)$ among all sites. Site $S_i^*$ is the corresponding most critical site and flow set $\mathbf{F}_i^* = \{F_j \in \mathbf{F}_i \wedge [a_j, b_j] \subseteq I\}$ is the corresponding critical flow set.
\end{definition}

Based on the above definitions, the pruning process works iteratively as follows: in each iteration we search for the most critical time interval $I^*$. Once this interval has been found, we check if feasibility can be satisfied in this interval, i.e., whether or not $\delta(S_i^*, I^*) \le 1$. If not, we remove the flow $F_j \in \mathbf{F}_i^*$ which has the maximized $d_j / (a_j \sim b_j)$, meaning it contributes the most to the intensity of interval $I^*$. By removing this flow, the intensity of this interval will be reduced. We repeat the above process until $\delta(S_i^*, I^*) \le 1$ is satisfied.

\begin{algorithm}[!t]
\caption{\label{alg:mcf-edf} \textbf{MCF-EDF}}

\begin{algorithmic}[1]
\WHILE{$\exists S_i \in \mathbf{S}, \mathbf{F}_i \neq \emptyset$}
	\STATE \emph{\textcolor{gray}{// Search for the most critical interval}}
	\STATE $(S_i^*, I^*) \leftarrow \argmax_{(S_i, [a, b])}{\{\delta(S_i, [a, b])\}}$
	\STATE $\mathbf{F}_i^* = \{ F_j~|~F_j \in \mathbf{F}_i \wedge [a_j, b_j] \subseteq [a, b] \}$
	\STATE \emph{\textcolor{gray}{// Schedule the flows using EDF}}
	\FOR{$F_j \in \mathbf{F}_i^*$}
		\STATE $a_j' = \sum_{F_k \in \mathbf{F}_i^* \wedge b_k < b_j}|F_k|$
		\STATE $b_j' = \sum_{F_k \in \mathbf{F}_i^* \wedge b_k < b_j}|F_k| + |F_j|$
	\ENDFOR
	\STATE \emph{\textcolor{gray}{// Update the available time on the affected sites}}
	\FOR{$F_j \in \mathbf{F}_i$}
		\STATE $\mathbf{F_i} \leftarrow \mathbf{F_i} \backslash F_j$ for $F_j \in \mathbf{S_i}$
		\STATE Mark $[a_j', b_j']$ as unavailable on site $S_i$ if $F_j \in \mathbf{F}_i$
	\ENDFOR
\ENDWHILE
\end{algorithmic}
\end{algorithm}

\bpara{Most Critical First with EDF.}We now discuss how to decide the schedule for the rest flows. We design an algorithm MCF-EDF (Most Critical First with EDF) for flow scheduling, which is listed in Algorithm~\ref{alg:mcf-edf}. The algorithm is conducted on an iterative process: in each iteration, MCF-EDF first finds the most critical interval $I^*$ and its corresponding most critical site $S_i^*$. The flows that fall into interval $I^*$ is denoted by set $\mathbf{F}_i^* = \{ F_j~|~F_j \in \mathbf{F}_i \wedge [a_j, b_j] \subseteq [a, b] \}$. After that, we schedule the flows in the interval $I^*$ using the EDF policy, from which the spanning time $[a_j', b_j']$ of each flow $F_j \in \mathbf{F}_i^*$ will be determined. Finally, we update the available time intervals on all the sites that have been affected, i.e., sites that contain flows from set $\mathbf{F}_i^*$ that have just been scheduled in this iteration. The algorithm terminates when all the flows have been scheduled.

\begin{theorem}
The schedule generated by MCF-EDF will guarantee that all the residual flows after the pruning process can meet their deadlines.
\end{theorem}

\begin{proof}
We prove it by contradiction. Assume there exists a flow $F^0$ whose deadline cannot be met and the time interval this flow has been scheduled is denoted by $I^0$. Without loss of generality, we assume all the other flows in $I^0$ can meet their deadlines. We denote by $d^0$ and $t^0$ the data volume and the amount of time assigned to flow $F^0$, respectively. We also denote by $\mathbf{F}^0$ the set of jobs that falls in interval $I^0$. According to the logic of the algorithm MCF-EDF, for all the other flows in $\mathbf{F}^0$ we have 
\begin{equation}
\frac{\sum_{F_j \in \{\mathbf{F}^0 \backslash F^0\}} d_j}{\sum_{F_j \in \{\mathbf{F}^0 \backslash F^0\}} t_j} = 1
\end{equation}
where $d_j$ and $t_j$ represents the data volume and the assigned amount of time for flow $F_j$. Consequently, combining with inequality $d^0/t^0 > 1$ due to the violation of the deadline of flow $F^0$, we can derive that
\begin{equation}
\frac{\sum_{F_j \in \mathbf{F}^0} d_j}{\sum_{F_j \in \mathbf{F}^0} t_j} > 1
\end{equation}
which contradicts the fact that no such interval will exist after the pruning process carried out before applying MCF-EDF. That completes the proof.
\end{proof}

\section{Experiments}
\label{sec:exp}

\begin{table}
\caption{\label{tb:para}Parameter Settings}
\centering
\begin{tabular}{ r|l }
  \hline
  \multicolumn{2}{c}{\textbf{Parameter Settings for DCs}} \\
  \hline
  Processing capacity $C_i$ & $\mathcal{U}(1, 9)$ \\
  Upload network bandwidth $B_i^{out}$ & $\mathcal{U}(1, 5)$ \\
  Download network bandwidth $B_i^{in}$ & $\mathcal{U}(1, 10)$ \\
  Electricity price $P_i$ & $\mathcal{N}(10, 3)$ \\
  Upload network price $Q_i^{out}$ & $\mathcal{N}(10, 3)$ \\
  Download network price $Q_i^{in}$ & $\mathcal{N}(5, 3)$ \\
  \hline
  \hline
  \multicolumn{2}{c}{\textbf{Parameter Settings for Jobs}} \\
  \hline
  Job arrival time and deadline $a_j, b_j$ & $\mathcal{U}(1, 100)$ \\
  Volume of input data $d_j$ & $\mathcal{N}(10, 5)$ \\
  Computation workload $l_j$ & $\mathcal{N}(6, 5)$ \\
  \hline
\end{tabular}
\end{table}

We validate the performance of the proposed assignment and scheduling algorithms by numerical simulations and present some preliminary results in this section.

We developed a discrete-time multi-data center simulator in Python, with the proposed assignment and scheduling algorithms implemented. The simulator exposes interfaces for various parameters in our model for both the data center and the job and the values for those parameters were generated randomly following the distributions summarized in Table~\ref{tb:para}. The distributions chosen for those parameters here are only based on experience and they serve only as a part of the primitive evaluation. Real values can be obtained or estimated in real-world implementations. The initial placement of the data set for each job to DC site is accomplished uniformly at random. 

We compare the proposed algorithm with a baseline approach. The baseline is defined as a greedy process in which jobs are assigned to the same DC site where the associated data resides. The scheduling of the jobs is done following a FCFS manner complemented with EDF for jobs arrived at the same moment, which is considered to be the de facto scheduling algorithm used in current cloud systems. We focus on two aspects of interest: admission rate and total cost. We notice that the optimal assignment we proposed can lead to few hotspot sites in the system, thought it produces the best cost. Alternatively, we make a small adjustment where for each job $J_j \in \mathbf{J}$, instead of choosing the site with the best cost, we choose a site uniformly at random from the candidate site set $\mathbf{S}_j$. This is a tiny trick made for a better tradeoff between admission rate and total cost, meaning that we compromise a bit gain on cost reduction but increase potentially the chance to achieve a better admission rate.

\begin{figure}[t!]
    \centering
    \subfigure[Small scale with 20 DC sites]{
        \centering
        \includegraphics[scale=0.42]{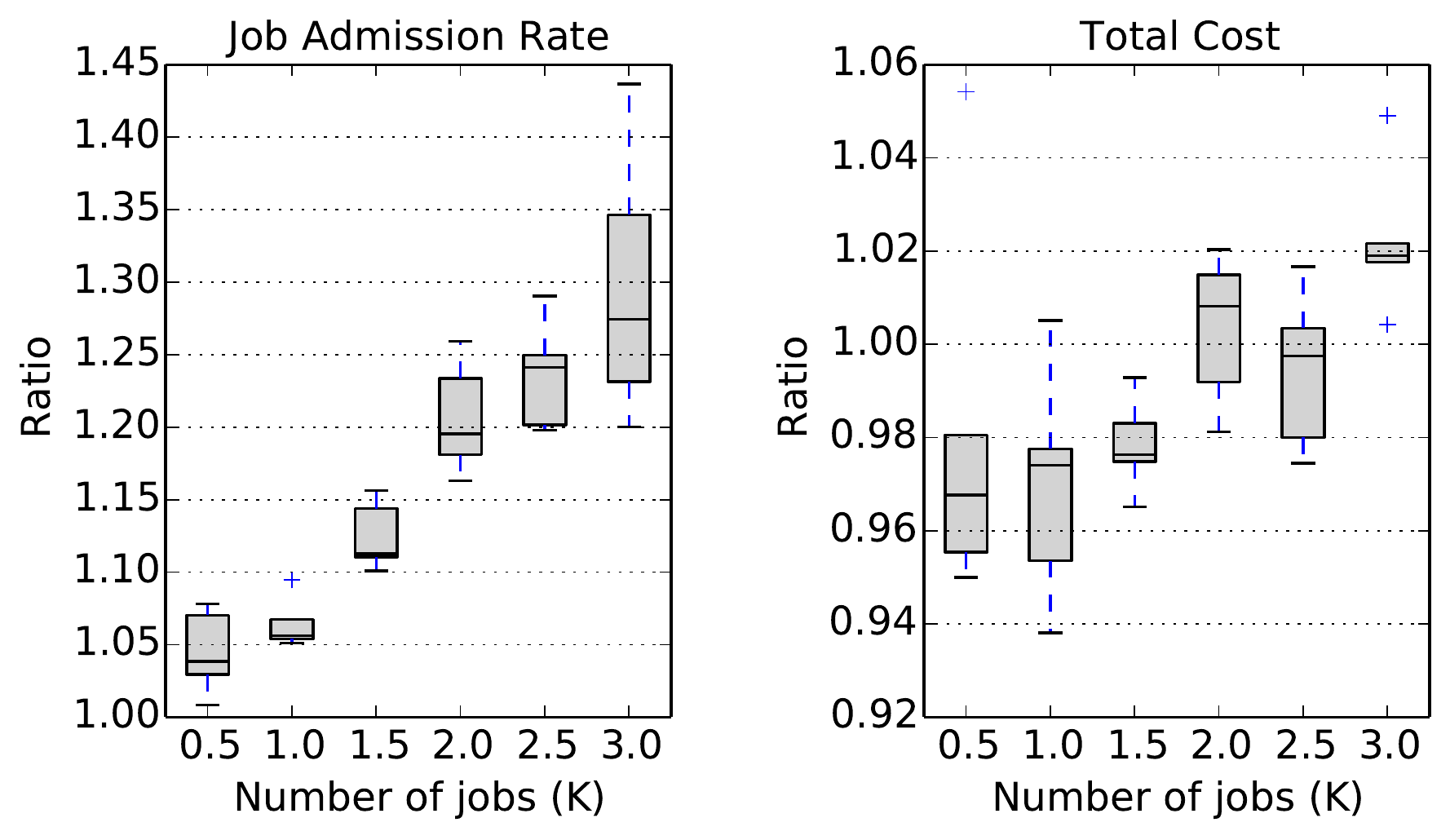}}
    \hspace{0.7cm}
    \subfigure[Large scale with 100 DC sites]{
        \centering
        \includegraphics[scale=0.42]{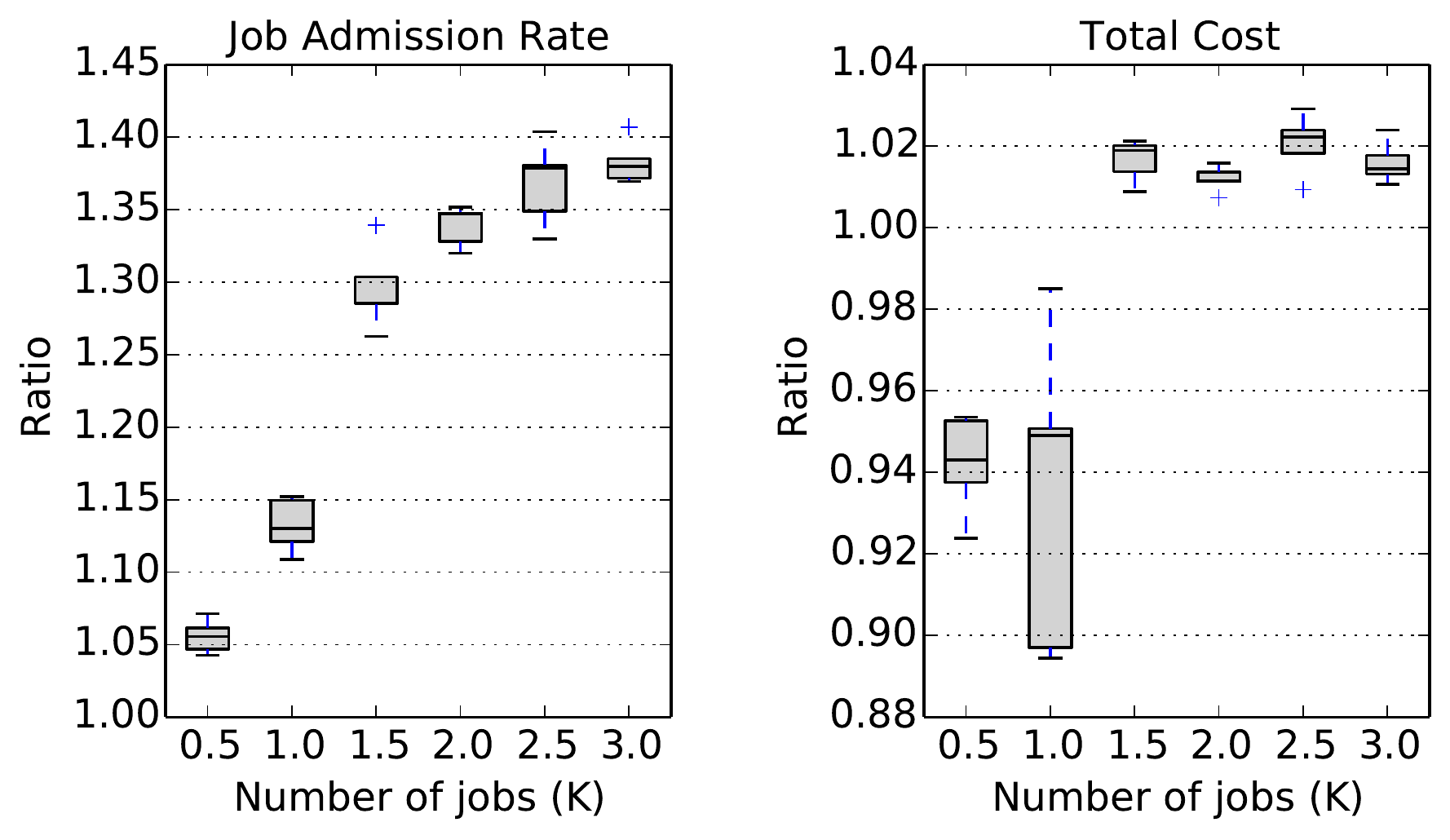}}
    \caption{\label{fig:results}Comparisons on job admission rate and total cost.}
\end{figure}

The experimental results are shown in Fig~\ref{fig:results}. All the values are normalized by the ones we obtained from the baseline approach and the values are averaged over five independent runnings. We tested with both a small scale (with 20 DC sites) and a large scale (with 100 DC sites) scenarios.  Experiments in both scenarios confirm that ($\mathit{i}$) the admission rate is enlarged in our approach and ($\mathit{ii}$) the proposed assignment and scheduling algorithm together can help reduce largely the energy cost. The former is clearly demonstrated on the figures while the latter follows by the fact that the total cost retains at more or less the same level while the admission rate is significantly improved. This is mainly due to the fact that the assignment and scheduling algorithm expand the solution space by enabling data transfer between DC pairs and by introducing elaborate scheduling mechanisms for joint data transfer and job execution.

\section{Discussion}
\label{sec:disc}

\bpara{Trade-offs between Admission Rate and Cost Efficiency.}
Admission rate and cost efficiency are two contradictory goals in our model. Reducing total cost means that more jobs will be transferred to DC sites with cheaper prices. However, this will lead to unexpected congestion on those favorable sites, resulting in a low admission rate. The trick of randomly choosing a site from multiple candidate sites in our experiments is a primary attempt to solving this problem. To exploit the right balance between the two goals, job assignment has to be done in a more sophisticated manner, e.g., based on the cost distribution of the sites for each job. 

\bpara{Privacy.}
Privacy concerns (such as in the EU \cite{EU-privacy-2014}) may result in more regulatory constraints on data movement. Our model can incorporate this scenario easily by introducing extra constraints for job assignment, while the computational complexity of the problem remains the same. In the sequel, the designed algorithm can still be applied, with possible adjustment such as pruning the candidate site set for each job to remove the unsatisfied sites due to data privacy constraints.

\bpara{Geo-distributed Processing and Analytics.} For simplicity the model we used in this paper assumes that the data for a job only resides at a single data center. However, this assumption does not necessarily hold in reality. For the scenarios where data can be spanning across multiple data centers, instead of gathering the data from all the data centers to a central point, several geo-distributed analytical frameworks \cite{Pu-GEO-2015, Vulimiri-Iridium-2015} have been recently proposed, where data movement is minimized or is subject to regulatory constraints. While it is generally achievable for private clouds as the management authority is owned by a single entity, applying the same technique in the public cloud domain is not trivial. We believe that cloud brokering or federation also brings opportunity for geo-distributed data analytics for public clouds. The most straightforward first step would be to develop a supporting framework for cloud brokers to provider added-value geo-distributed data processing services for the tenants. We leave this line of research for future work.

\section{Conclusions}
\label{sec:conc}

This paper studies the problem of achieving energy efficiency in public cloud systems from cloud broker's point of view. We provided the formulation of the problem as well as necessary analysis on its computational complexity. By introducing admission control, we simplified the problem and proposed efficient algorithm for the resulting problem of scheduling jobs for maximized admission rate. The job scheduling consists in two coupled phases namely data transfer and job execution, thus the proposed algorithm first decouples the two phases and then, it devotes to efficient scheduling for each of the phases. Experiments in both small and large scales confirmed the hypothesis that considerable reduction on energy consumption can be achieved through elaborate job assignment and scheduling by the broker in public cloud systems.

\bibliographystyle{abbrv}
\bibliography{refs}  

\balance

\end{document}